\documentclass[conference,letterpaper]{IEEEtran}

\usepackage{amsmath,amssymb,amsfonts,amsthm}
\usepackage{graphicx}
\usepackage{booktabs}
\usepackage{multirow}
\usepackage{enumitem}
\usepackage{cite}
\usepackage{comment}
\usepackage[textsize=tiny]{todonotes}

\theoremstyle{plain}
\newtheorem{theorem}{Theorem}
\newtheorem{lemma}{Lemma}
\newtheorem{corollary}{Corollary}

\theoremstyle{remark}
\newtheorem{remark}{Remark}
\newtheorem{example}{Example}

\newcommand{\Ftwo}{\mathbb{F}_2}
\newcommand{\wt}{\mathrm{w}}

\newcommand{\T}{\mathsf{T}}

\usepackage[acronym,shortcuts]{glossaries}

\newacronym{AWGN}{AWGN}{additive white Gaussian noise}
    \newacronym{BPSK}{BPSK}{binary phase-shift keying}
\newacronym{LDPC}{LDPC}{low-density parity-check}
\newacronym{PRC-LDPC}{PRC-LDPC}{primitive rate-compatible LDPC}
\newacronym{QC}{QC}{quasi-cyclic}

\usepackage{tikz}
\usetikzlibrary{arrows.meta,automata, positioning}
\usetikzlibrary{calc, chains,
                fit,
                positioning,
                shapes}
\usetikzlibrary{patterns}       
\usepackage{arydshln}

\usepackage{pgfplots}
\usepackage{dashrule}

\usetikzlibrary{matrix,fadings,decorations.pathmorphing,intersections}
\usepgfplotslibrary{fillbetween}
\pgfdeclarelayer{bg}
\pgfsetlayers{bg,main}
\makeatletter
\tikzset{use path/.code=\tikz@addmode{\pgfsyssoftpath@setcurrentpath#1}}
\makeatother

\begin{document}

\title{Design of a Family of Rate-Adaptive LDPC Convolutional Codes with Short Constraint Length}

\author{\IEEEauthorblockN{Massimo Battaglioni, Edoardo Carnevali, Franco Chiaraluce, Marco Baldi, Giovanni Cancellieri}\\
\IEEEauthorblockA{\textit{Department of Information Engineering}, Università Politecnica delle Marche\\
Ancona, 60131, Italy \\
\{m.battaglioni, f.chiaraluce, m.baldi, g.cancellieri\}@univpm.it, edoardo.carnevali@pm.univpm.it}
}

\maketitle

\begin{abstract}
This paper studies the design of some high-rate, adaptive, low-density parity-check (LDPC) convolutional codes with short constraint length, based on interleaved parity-check polynomials.
Starting from some known results regarding the selection of polynomials that yield codes free of short cycles in the associated Tanner graphs, we propose some new constructions based on Golomb rulers, along with relaxed variants that do not strictly satisfy the Golomb ruler condition, but retain similar structural properties. The result is a family of rate-adaptive code constructions with code rates of the form $R=(n_0-1)/n_0$. Conditions for the absence of length-4 cycles are derived, together with upper bounds on the free distance that are met in all considered examples.
\end{abstract}


\section{Introduction}

\Ac{LDPC} convolutional codes are placed at the intersection of convolutional codes and \ac{LDPC} codes, and have attracted growing interest in recent years.
They are usually represented through their parity-check matrix \cite{FelteyLDPCConv}, which has some special properties. At each time instant $t$, a convolutional encoder maps an input block of $k_0$ information bits into a block of $n_0>k_0$ coded bits as a function of the current and a finite number of past input blocks, yielding a sequential code with asymptotic rate $R=k_0/n_0$. Termination can be obtained by ending the sequence of coded blocks after a finite number of time instants while retaining the parity constraints induced by the final blocks. This slightly reduces the effective rate, although the loss vanishes for long frames and can be completely avoided through tail biting~\cite{TailBiting}. After suitable row and column permutations, the parity-check matrix of a tail-biting LDPC convolutional code can be written in \ac{QC} form~\cite{QCcodes}, that is, as a matrix composed of circulant blocks.
Several design approaches start from such a form, optimize the choice of the circulant blocks and then unwrap the matrix into a sparse convolutional representation~\cite{Pusane2011,Battaglioni2023Mitch,Yang2023}. In contrast, the approach we propose in this work does not rely on unwrapping, but instead obtains the convolutional parity-check matrix via periodic row elimination applied to the parity-check matrix of a class of rate-compatible block \ac{LDPC} codes derived from cyclic codes. 
Rateless block codes based on primitive parity-check polynomials have been investigated in \cite{Shirvanimoghaddam}. Their rate-compatible \ac{LDPC} counterparts were later introduced in \cite{Battaglioni2024Early} under the name of \ac{PRC-LDPC} codes and further analyzed in \cite{BattaglionisimplexTIT2025, Namadchi}. The construction we propose stems from this line of research, but without requiring the underlying parity-check polynomials to be primitive. This relaxation significantly broadens the design space and differentiates our construction from existing approaches. To avoid harmful length-$4$ cycles in the code's Tanner graph, constructions inspired by Golomb rulers \cite{Golomb1972} are exploited.  Notably, unlike block codes, the convolutional framework we consider enables controlled relaxation of the Golomb ruler constraint, which leads to fine rate adaptivity.

In~\cite{Cho2015}, short-memory protographs free of length-$4$ cycles are obtained by imposing distinct intra-column distances, whereas our construction permits repeated differences and provides multiple rates from a single exponents vector. Rate adaptivity in~\cite{KoikeAkino2018RateAdaptive} is instead achieved by shortening and puncturing a mother \ac{LDPC} convolutional code.
Another well-known family with structural similarities to our construction is Massey's convolutional self-orthogonal codes \cite{Massey1963Threshold}. Since their introduction, several design methodologies have been developed \cite{Robinson1967RecurrentCodes, Klieber1970DifferenceTriangles, Wu1975ConvCodesPart1, Wu1976ConvCodesPart2, Haccoun2005SelfDoublyOrthogonal},  all relying on self-orthogonality conditions, typically expressed as parity constraints on shifted parity-check equations overlaps.
More recently, such constructions have been revisited to derive spatially coupled \ac{LDPC} codes via protograph lifting \cite{Costello2025HighRateSCLDPC}. In contrast, our approach is not based on generator-matrix orthogonality constraints and does not enforce self-orthogonality, but relies on a parity-check polynomial representation. In this setting, the cycle properties of the resulting Tanner graph are governed by modular conditions on the differences of the exponents vector, rather than by parity constraints on overlapping parity-check equations. This leads to a different design paradigm, enabling a larger set of admissible constructions while still allowing control of short cycles.
This enables the design of a family of high-rate \ac{LDPC} convolutional codes with adaptive rates of the form \(R=\frac{n_0-1}{n_0}\). A single interleaved parity-check polynomial \(h(x)\) and the block length \(n_0\) determine the parity-check matrix, enabling a direct characterization of the cycles and closed-form free-distance bounds~\cite{Cancellieri2015},  which coincide with the distance estimates in all considered examples.


\section{Notation and Background}
\label{sec:prelim}

We work over the binary field $\Ftwo$. A binary polynomial of degree $L-1$ is written as

\begin{equation}
h(x) = \sum_{j=0}^{L-1} h_j x^j,\qquad h_j \in \Ftwo,
\label{eq:poly}
\end{equation}
and has weight $\gamma=\wt[h(x)]$ equal to the number of its nonzero coefficients. We associate with \(h(x)\) its coefficient vector
\(\mathbf h=(h_0,\ldots,h_{L-1})\).  Equivalently, we represent a polynomial  by listing the exponents of its nonzero coefficients, e.g.,
\[
h(x)\;\leftrightarrow\; (a_0,a_1,\dots,a_{\gamma-1})\triangleq \mathbf{a},
\]
where $a_0<a_1<\cdots<a_{\gamma-1}=L-1$. In this paper, we assume $a_0=0$.  
We call $\mathbf{a}$ the \emph{exponents vector}. Thus, $L$ denotes the polynomial support length, whereas the associated Golomb ruler has length $a_{\gamma-1}=L-1$, with entries at positions $a_i$. The reciprocal of a polynomial $h(x)$ of degree $k$ is denoted as $h^*(x)=x^kh(x^{-1})$.

A binary linear block code $C$ of dimension $k$ is
\[
C=\{\mathbf c\in\mathbb F_2^n:\mathbf c\mathbf H^{\mathsf T}=\mathbf0\},
\]
that is, the kernel of a parity-check matrix $\mathbf H\in\mathbb F_2^{\rho\times n}$, for some $\rho\geq \operatorname{rank}(\mathbf H)=n-k$. Its rate is $R=k/n$, and its minimum distance $d_{\min}$ is the smallest Hamming weight of a nonzero codeword. For cyclic codes, we denote the codeword length by $N$. Given a parity-check polynomial $h(x)=\sum_{i=0}^k h_i x^i$ of degree $k$ over $\mathbb{F}_2$, with $h_0=h_k=1$, the parity-check matrix of the corresponding cyclic code is 

\begin{equation}
 \widetilde{\mathbf{H}} =
\begin{tikzpicture}[baseline=(current bounding box.center) ]
 \matrix[matrix of math nodes,left delimiter=(,right delimiter=),scale=0.6] (m)
   {
h_0 & h_1 & \cdots & h_k & & &&&& \\
& h_0 & h_1& \cdots  & h_k &&&&& \\
& & \ddots & \ddots & \vdots &&& \\
& &  & h_0 & h_1 & \cdots & h_k &&\\
};
\draw[<->,thick] ([xshift=5mm]m.south east) -- node[right] {\begin{tabular}{@{}l@{}}
    $N-k$, 
  \end{tabular}} ([xshift=5mm]m.north east);
\end{tikzpicture}
\label{eq:formH_simplexcodes}
\end{equation}
where blank spaces represent zeros.

A time-invariant convolutional code can be defined as the set of code sequences $\mathbf v=\{\mathbf v_t\}_{t\geq0}$, with $\mathbf v_t=\mathbf 0$ for $t<0$, satisfying a linear finite-memory parity constraint described by the polynomial matrix
\[
H^{\T}(D)=H_0^{\T}+H_1^{\T}D+\cdots+H_{m_s}^{\T}D^{m_s},
\]
called the syndrome former. Let $\mathbf v_t$ denote the block of $n_0$ symbols transmitted
at time $t$. Writing $\mathbf v_t=(v_t^{(0)},\ldots,v_t^{(n_0-1)})$, we associate with a finite-weight code sequence the scalar polynomial
\[
v(x)=\sum_{t\geq0}\sum_{r=0}^{n_0-1}v_t^{(r)}x^{n_0t+r}.
\] The parity relations are
\begin{equation}
\mathbf v_t H_0^{\T}
+\mathbf v_{t-1} H_1^{\T}
+\cdots
+\mathbf v_{t-m_s} H_{m_s}^{\T}
=\mathbf 0,
\qquad t\geq 0.
\label{eq:time_syndrome}
\end{equation}
Equivalently, \(    \mathbf v\,\mathbf H^{\T}=\mathbf 0,\)
where the semi-infinite matrix
\begin{equation}
\mathbf H^{\T}=
\begin{pmatrix}
H_0^{\T} & H_1^{\T} & \cdots & H_{m_s}^{\T} &        &        & \\
         & H_0^{\T} & H_1^{\T} & \cdots & H_{m_s}^{\T} &        & \\
                  &          & \ddots & \ddots & \ddots & \ddots
\end{pmatrix}
\label{eq:syndromeformer}
\end{equation}
is the semi-infinite syndrome former matrix. It is convenient to define the \emph{syndrome former submatrix} \(\mathbf{H}_s=\begin{pmatrix}
    H_0^{\T} & \cdots & H_{m_s}^{\T} \\
\end{pmatrix}\). 
The integer $m_s$ is the syndrome former memory order and the syndrome former constraint length is defined as $\nu_s=(m_s+1)n_0$. The \emph{free} distance $d$ of a convolutional code is the smallest Hamming weight among all nonzero finite-weight code sequences. For soft-decision maximum-likelihood decoding, the asymptotic coding gain is $G_\infty=Rd$~\cite[Chapter 1]{LinCost}. A convolutional code is said to be an \ac{LDPC} convolutional code if the column weight of its parity-check matrix is much smaller than $(n_0-k_0)(m_s+1)$.

\remark Starting from  \eqref{eq:formH_simplexcodes},  if \(
N-k\ge n_0\left(\left\lceil\frac{L}{n_0}\right\rceil-1\right)+1
\), a convolutional
structure can be obtained by periodic row elimination: keeping one row every $n_0$ produces a finite syndrome former. The repetition of this submatrix along time generates the  convolutional parity relations. The corresponding convolutional code has  $R=\frac{n_0-1}{n_0}$ and $\nu_s=n_0\left\lceil \frac{L}{n_0} \right\rceil$.  \label{rem:decim}

Given the target block length $n_0$ and rate $(n_0-1)/n_0$ of a convolutional code we aim at designing, it is useful to partition its parity-check polynomial $h(x)$ into $n_0$ sub–polynomials $h_r(\cdot)$ with smaller support as follows
\begin{equation}
h(x)=\sum_{r=1}^{n_0} x^{n_0-r}h_r(x^{n_0}),
\label{eq:mod-decomp}
\end{equation}
so that each sub–polynomial describes the contribution of the $r$-th symbol inside each block of $n_0$ symbols. Accordingly, $h(x)$ is called \emph{interleaved} parity-check polynomial. 

\begin{example}
Consider the  parity-check polynomial
\[
h(x)=1+x^2+x^5 \quad\Longleftrightarrow\quad \mathbf{a}=(0,2,5),
\]
and set \( n_0 = 2 \). By removing one row every \( n_0=2 \) rows from \eqref{eq:formH_simplexcodes}, and replicating the resulting matrix, we obtain a time-invariant convolutional code with $m_s=2$ and $\nu_s=6$. The transposed syndrome former submatrix is 
\[
\mathbf H_s^{\mathsf T} =
\begin{pmatrix}
0 & 1 \\
1 & 0 \\
1 & 0
\end{pmatrix}.
\]
\label{ex:ex1}
\end{example}

A length-$4$ cycle in the Tanner graph of $\mathbf H$ corresponds to a $2\times 2$ submatrix of $\mathbf H$ whose entries are all equal to $1$. A length-$4$ cycle occurs in \eqref{eq:formH_simplexcodes} if, given $h(x)$ and the corresponding $\mathbf a$, there exist two distinct ordered pairs $(i,j)\neq(i',j')$, with $0\leq i<j\leq\gamma-1$ and $0\leq i'<j'\leq\gamma-1$, such that~\cite{BattaglioniFitce}:
\begin{equation}
S_{j,i}=S_{j',i'}>0,
\label{eq:n2_equal_diff}
\end{equation}
where, for any  $(i,j)$ with $0\le i<j\le \gamma-1$, the span $S_{j,i}$ is defined as
\begin{equation}
S_{j,i}\triangleq a_j-a_i .
\label{eq:suv}
\end{equation}
We refer to an equality of the form
\(
S_{j,i}=S_{j',i'}
\)
between two distinct ordered pairs as an \emph{equal-span relation}.

\remark Following \eqref{eq:n2_equal_diff}, a necessary and sufficient condition to avoid length-$4$ cycles in a \ac{PRC-LDPC} code \cite{Battaglioni2024Early} is that the exponents vector $\mathbf{a}$ is a Golomb ruler. \label{rem:rem2}

\section{LDPC convolutional codes with $n_0=2$}
\label{sec:half_rate}

Let \(h(x)\) be an interleaved parity-check polynomial of weight $\gamma=\wt[h(x)]$, as in \eqref{eq:poly}. For $n_0=2$, the decomposition \eqref{eq:mod-decomp} becomes
\begin{equation}
h(x)=x\,h_1(x^2)+h_2(x^2),
\label{eq:n2_decomposition}
\end{equation}
where $h_1(\cdot)$ collects coefficients at odd positions and $h_2(\cdot)$ those at even positions.

Let us associate to each ordered pair $(u,v)$, with $0 \le u < v \le \gamma-1$, the interval \(\mathcal I(u,v) \triangleq [a_u,a_v]\) of length $S_{v,u}$. Let $\mathcal I(i,j)$ and $\mathcal I(i',j')$ be two such intervals with $i<j$ and $i'<j'$, and assume, without loss of generality, that $a_i < a_{i'}$. Define the (signed) separation
\begin{equation}
Q_{i',j} \triangleq a_{i'} - a_j .
\label{eq:def_Q}
\end{equation}
Thus $Q_{i',j}=0$ means that the two intervals are adjacent, $Q_{i',j}>0$ means that they are disjoint, while $Q_{i',j}<0$ means that they overlap. The quantity $Q_{i',j}$ is useful because the same exponents involved in an equal-span relation may generate a second repeated span, as characterized next. Let us also define 
\[
W_{n_0}(r)\triangleq\#\{a_k: a_k\equiv r \pmod{n_0}\}.
\]
Thus, with the columns indexed by $r=0,\ldots,n_0-1$, $W_{n_0}(r)$ is the weight of the $r$-th column of $\mathbf H_s^{\mathsf T}$, and $\sum_{r=0}^{n_0-1}W_{n_0}(r)=\gamma$.

\begin{lemma}
\label{lem:adjacency_classes}
Given an exponents vector $\mathbf a$, assume that \eqref{eq:n2_equal_diff} holds for two distinct ordered pairs $(i,j)\neq(i',j')$, with $i<j$ and $i'<j'$. After exchanging the two pairs if necessary, assume that $a_i<a_{i'}$. Then, the corresponding entries can generate equal-span relations with at most two distinct span values: $S_{j,i}$ and, if $Q_{i',j}\neq0$, $S_{j,i}+Q_{i',j}$. When $Q_{i',j}=0$, only the span $S_{j,i}$ occurs.
\end{lemma}
\begin{proof}
Let $S=S_{j,i}=S_{j',i'}$. Since $a_j=a_i+S$ and $a_{j'}=a_{i'}+S$, using $Q_{i',j}=a_{i'}-a_j$ gives
\[
a_{i'}-a_i=a_{j'}-a_j=S+Q_{i',j}.
\]
Hence, besides the original repeated span $S$, the only possible distinct repeated span generated by these entries is $S+Q_{i',j}$. When $Q_{i',j}=0$, the two span values coincide.
\end{proof}

\begin{lemma}
\label{lem:n2_survival}
Under periodic row elimination with period equal to the block length \(n_0=2\), a length-$4$ cycle with span \(S\) survives in the parity-check matrix of the resulting \ac{LDPC} convolutional code if and only if \(S \equiv 0 \pmod{2}\).
\end{lemma}
\begin{proof}
Under periodic row elimination with period \(n_0=2\), the rows of
\(\mathbf{H}\) are partitioned into even-indexed and odd-indexed rows, and only one class is retained. A length-$4$ cycle with span \(S\) involves two rows whose indexes differ by exactly \(S\). Both rows belong to the same parity class if and only if their index difference \(S\) is divisible by \(2\), i.e., \(S\equiv 0\pmod{2}\). If they belong to different parity classes, one row is removed by the elimination and the cycle is destroyed. If they belong to the same class, both rows are retained and the cycle survives.
\end{proof}

\begin{theorem}
\label{thm:n2_cycle_survival}
Given a parity-check polynomial $h(x)$ with exponents vector $\mathbf{a}$, suppose \eqref{eq:n2_equal_diff} holds for some indexes $i<j$ and $i'<j'$.
After row elimination with $n_0=2$:
\begin{enumerate}
    \item If $Q_{i',j}=0$ (adjacent intervals), the length-$4$ cycles induced by the considered configuration have span $S_{j,i}$ and survive if and only if $S_{j,i}$ is even.
    \item If $Q_{i',j}\neq 0$ (non-adjacent intervals), the length-$4$ cycles induced by the considered configuration can have spans $S_{j,i}$ and $S_{j,i}+Q_{i',j}$. The cycles associated with either span survive if and only if that span is even.
\end{enumerate}
\end{theorem}
\begin{proof}
By Lemma~\ref{lem:adjacency_classes}, the cycles induced by the considered configuration can have span $S_{j,i}$ and, when $Q_{i',j}\neq0$, $S_{j,i}+Q_{i',j}$. The survival criterion of Lemma~\ref{lem:n2_survival} applies to each span, giving the stated conclusions.
\end{proof}

\begin{corollary}
\label{cor:n2_compact}
Under the assumptions of Theorem~\ref{thm:n2_cycle_survival}, if \(S_{j,i}\) is odd and \(Q_{i',j}\) is even, then none of the length-$4$ cycles induced by \eqref{eq:n2_equal_diff} survives row elimination with \(n_0=2\).
\end{corollary}

\begin{remark}
The previous results show that, when constructing an \ac{LDPC} convolutional code as in Remark \ref{rem:decim}, requiring \(\mathbf{a}\) to be a Golomb ruler is sufficient but not necessary to eliminate all length-\(4\) cycles.
\end{remark}

Let us recall a result from \cite{Cancellieri2015} on the free distance of convolutional codes, which easily extends to the \ac{LDPC} case. 

\begin{theorem}
\label{thm:basic_codeword}
The polynomial \(v(x)=h^{*}(x) \) is always a codeword of any rate-\(1/2\) convolutional code constructed as described in Remark \ref{rem:decim}. Therefore, the free distance satisfies
\begin{equation}
d \le \wt[h(x)].
\end{equation}
\end{theorem}

\section{LDPC convolutional codes  with $n_0\geq 2$}
\label{sec:high_rate}
As noted in Remark~\ref{rem:decim}, starting from a cyclic code,  one can obtain a convolutional code with rate \(R=\frac{n_0-1}{n_0}\) by periodically eliminating \(n_0-1\) consecutive rows out of every block of \(n_0\) rows.

\begin{remark}
We focus on nondegenerate constructions in which all $n_0$ parity-check sub-polynomials are nonzero. If one sub-polynomial is zero, the corresponding symbol position is unconstrained, and a weight-$1$ codeword exists. Hence, the resulting convolutional code has free distance $d=1$.
\end{remark}

\begin{theorem}
\label{thm:adjacent_modn0}
Let \(h(x)\) be a parity-check polynomial, and let \(\mathbf a\) be the corresponding exponents vector.  Assume that there exist indexes \(i<j\) and \(i'<j'\) such that
\[
S_{j,i} = S_{j',i'} > 0, \qquad Q_{i',j}=0.
\]
Then, the corresponding length-\(4\) cycle in the parent parity-check matrix \eqref{eq:formH_simplexcodes} survives in the parity-check matrix of the convolutional code obtained by periodic row elimination with period \(n_0\) if and only if
\begin{equation}
    S_{j,i} \equiv 0 \pmod{n_0}.
    \label{eq:condizionemod}
\end{equation}
\end{theorem}

\begin{proof}
By Lemma~\ref{lem:adjacency_classes}, adjacent intervals yield a single equal-span relation with span \(S_{j,i}\). Under period-\(n_0\) elimination, two rows remain aligned across periods if and only if their row-index difference is a multiple of \(n_0\), i.e., \(S_{j,i} \equiv 0 \pmod{n_0}\). If \(S_{j,i} \not\equiv 0 \pmod{n_0}\), the cycle cannot occur.
\end{proof}

\begin{theorem}
\label{thm:separated_modn0}
Let $h(x)$ be a parity-check polynomial, and let $\mathbf a$ be the corresponding exponents vector. Assume that there exist indexes $i<j$ and $i'<j'$ such that
\[
S_{j,i}=S_{j',i'}>0,\qquad Q_{i',j}\neq0.
\]
Then, by Lemma~\ref{lem:adjacency_classes}, the length-$4$ cycles induced by the considered configuration can have spans $S_{j,i}$ and $S_{j,i}+Q_{i',j}$. After periodic row elimination with period $n_0$, a cycle survives if and only if its span is a multiple of $n_0$. Therefore, cycles with span $S_{j,i}$ survive when
\[
S_{j,i}\equiv0\pmod{n_0},
\]
whereas those with span $S_{j,i}+Q_{i',j}$ survive when
\[
S_{j,i}+Q_{i',j}\equiv0\pmod{n_0}.
\]
\end{theorem}

\begin{proof}
The possible span values follow from Lemma~\ref{lem:adjacency_classes}. The result then follows from the period-$n_0$ row-elimination rule.
\end{proof}

\begin{corollary}
\label{cor:separated_modn0_compact}
Under the hypotheses of Theorem~\ref{thm:separated_modn0}, no length-$4$ cycle induced by the considered repeated-difference configuration survives if and only if
\begin{equation}
S_{j,i}\not\equiv0\pmod{n_0},\qquad
Q_{i',j}\not\equiv-S_{j,i}\pmod{n_0}.
\label{eq:cor_compact_modn0}
\end{equation}
\end{corollary}

Let us show that a Golomb-ruler construction removes all length-$4$ cycles, but may still produce a degenerate code for some values of $n_0$.

\begin{example}
\label{ex:golo_base}
Consider the following exponents vector with $\gamma=7$ and $L=26$, whose entries form a Golomb ruler:
\[
\mathbf a=(0,1,4,10,18,23,25).
\]
 All differences are distinct; hence, no length-4 cycle exists in the convolutional parity-check matrix for any \(n_0\). However, for $n_0=5$, $\mathbf H_s^{\mathsf T}$ has an all-zero column, since $W_5(2)=0$. The  corresponding rate-$4/5$ code is therefore degenerate.

\end{example}

Let us now show that a small perturbation in the assumed entries may introduce repeated differences, while the resulting length-$4$ cycles can still be removed by periodic row elimination.

    \begin{example}
\label{ex:adjacent_equal}
Replace \(10\) with \(12\) in Example~\ref{ex:golo_base}:
\[
\mathbf a=(0,1,4,12,18,23,25).
\]
A repeated difference appears:
\[
S_{3,1}=a_3-a_1=12-1=11, \; S_{5,3}=a_5-a_3=23-12=11,
\]
hence \(\mathbf a\) is not a Golomb ruler. Moreover, \(Q_{3,3}=a_3-a_3=0,\) so the corresponding intervals are adjacent. Therefore, by Theorem~\ref{thm:adjacent_modn0}, a length-$4$ cycle survives after periodic row elimination with period \(n_0\) if and only if \(S_{3,1} \equiv 0 \pmod{n_0}.\) 
Since \(S_{3,1}=11 \not\equiv 0 \pmod{n_0}\) for \(n_0\in\{2,3,4,5\},\) no length-$4$ cycle survives in the convolutional parity-check matrix for any of the considered periods.
\end{example}

Let us then show how adjacent and non-adjacent repeated differences lead to different  conditions for cycle survival.

\begin{example}
\label{ex:separated_equal}
Replace $10$ with $8$ in Example~\ref{ex:golo_base}:
\[
\mathbf a=(0,1,4,8,18,23,25).
\]
The adjacent repeated difference
\[
S_{2,0}=S_{3,2}=4,\qquad Q_{2,2}=0,
\]
induces a cycle that survives for $n_0\in\{2,4\}$. The non-adjacent repeated difference
\[
S_{3,1}=S_{6,4}=7,\qquad Q_{4,3}=10,
\]
induces  cycles with spans $7$ and $17$, neither of which survives for any $n_0\in\{2,3,4\}$. Hence, this configuration produces no additional surviving cycle. Moreover, all $W_{n_0}(r)$ are nonzero for $2\leq n_0\leq4$; so, the corresponding codes are nondegenerate.
\end{example}

\subsection{Free distance properties}

We now extend the upper bound on the free distance $d$ provided in Theorem \ref{thm:basic_codeword} to convolutional codes having rate $R=\frac{n_0-1}{n_0}$, with $n_0>2$. Consider the syndrome-former submatrix $\mathbf H_s^{\mathsf T}$ associated with a convolutional code constructed as in Remark~\ref{rem:decim}. Consider, for any pair $0\leq i<j\leq n_0-1$, code sequences for which $v_t^{(r)}=0$ for every $r\notin\{i,j\}$ and every $t$. The parity-check equations then involve only the $i$-th and $j$-th columns of $\mathbf H_s^{\mathsf T}$, which define a rate-$1/2$ convolutional code.

Let us define
\begin{equation}
\sigma_{i,j}\triangleq W_{n_0}(i)+W_{n_0}(j),
\qquad 0\leq i<j\leq n_0-1.
\label{eq:pairwise_sums}
\end{equation}

\begin{theorem}
\label{thm:distance_bound_general}
Let $\mathcal C$ be a convolutional code constructed as in Remark~\ref{rem:decim}, with rate $R=\frac{n_0-1}{n_0}$, $n_0>2$, and let $\mathbf H_s^{\mathsf T}$ be its syndrome-former submatrix. Then, the free distance $d$ of $\mathcal C$ satisfies
\begin{equation}
d\leq\min_{0\leq i<j\leq n_0-1}\sigma_{i,j}.
\label{eq:distance_bound_general}
\end{equation}
\end{theorem}
\begin{proof}
Fix a pair $0\leq i<j\leq n_0-1$. Restricting the code sequences to $v_t^{(r)}=0$ for all $r\notin\{i,j\}$ reduces the parity-check equations of $\mathcal C$ to those defined by the $i$-th and $j$-th columns of $\mathbf H_s^{\mathsf T}$. These two columns define a rate-$1/2$ convolutional code. By the construction in~\cite{Cancellieri2015}, this code contains a nonzero codeword of Hamming weight
\[
W_{n_0}(i)+W_{n_0}(j)=\sigma_{i,j}.
\]
Setting the remaining $n_0-2$ coordinates identically to zero yields a codeword of $\mathcal C$ with the same Hamming weight. Hence,
\(d\leq\sigma_{i,j}\). Since this holds for every pair $0\leq i<j\leq n_0-1$, minimizing over all pairs gives \eqref{eq:distance_bound_general}.
\end{proof}

\begin{table*}[t!]
\centering
\renewcommand{\arraystretch}{1.15}
\setlength{\tabcolsep}{5pt}
\caption{Examples~\ref{ex:golo_base}--\ref{ex:separated_equal}: syndrome-former constraint length, column-weight vectors of \(\mathbf H_s^{\mathsf T}\), and corresponding upper bounds on \(d\).}
\label{tab:examples}
\begin{tabular}{l c c c c c c c c c}
\hline
\multirow{2}{*}{Ref./Exponents vector} &
\multirow{2}{*}{\(R\)} &
\multirow{2}{*}{\(\nu_s\)} &
\multicolumn{5}{c}{Column weights \(W_{n_0}(r)\) of \(\mathbf H_s^{\mathsf T}\)} &
\multirow{2}{*}{Bound on \(d\)} &
\multirow{2}{*}{Estimated \(d\)}
\\
\cline{4-8}
& & & \(r=0\) & \(r=1\) & \(r=2\) & \(r=3\) & \(r=4\) & & \\
\hline

\multirow{4}{*}{\parbox{5.2cm}{\textbf{Example \ref{ex:golo_base}}\\
\(\mathbf a=(0,1,4,10,18,23,25)\)\\
(Golomb ruler)}}
& \(1/2\) & 26 & 4 & 3 &  &  &  & \(d\le 7\) & \(7\) \\
& \(2/3\) & 27 & 2 & 4 & 1 &  &  & \(d\le 3\) & \(3\) \\
& \(3/4\) & 28 & 2 & 2 & 2 & 1 &  & \(d\le 3\) & \(3\) \\
& \(4/5\) & 30 & \multicolumn{5}{c}{\emph{degenerate}} &  &  \\
\hline

\multirow{4}{*}{\parbox{5.2cm}{\textbf{Example \ref{ex:adjacent_equal}}\\
\(\mathbf a=(0,1,4,12,18,23,25)\)}}
& \(1/2\) & 26 & 4 & 3 &  &  &  & \(d\le 7\) & \(7\) \\
& \(2/3\) & 27 & 3 & 3 & 1 &  &  & \(d\le 4\) & \(4\) \\
& \(3/4\) & 28 & 3 & 2 & 1 & 1 &  & \(d\le 2\) & \(2\) \\
& \(4/5\) & 30 & 2 & 1 & 1 & 2 & 1 & \(d\le 2\) & \(2\) \\
\hline

\multirow{4}{*}{\parbox{5.2cm}{\textbf{Example \ref{ex:separated_equal}}\\
\(\mathbf a=(0,1,4,8,18,23,25)\)}}
& \(1/2\)\textsuperscript{\(\ast\)} & 26 & 4 & 3 &  &  &  & \(d\le 7\) & \(7\) \\
& \(2/3\) & 27 & 2 & 3 & 2 &  &  & \(d\le 4\) & \(4\) \\
& \(3/4\)\textsuperscript{\(\ast\)} & 28 & 3 & 2 & 1 & 1 &  & \(d\le 2\) & \(2\) \\
& \(4/5\) & 30 & \multicolumn{5}{c}{\emph{degenerate}} &  &  \\
\hline
\end{tabular}
\vspace{1mm}
\parbox{\textwidth}{\footnotesize
\textsuperscript{$\ast$}The parity-check matrix contains length-$4$ cycles.
}
\end{table*}

\begin{table*}[t!]
\centering
\renewcommand{\arraystretch}{1.15}
\setlength{\tabcolsep}{6pt}
\caption{Performance of the optimal $23$-entry Golomb ruler of length $372$ (Example~\ref{ex:large_golomb}).}
\label{tab:large_golomb}

\begin{tabular}{c c cccccc cccc}
\hline
\multirow{2}{*}{Rate \(R\)} &
\multirow{2}{*}{\(\nu_s\)} &
\multicolumn{6}{c}{Column weights \(W_{n_0}(r)\)} &
\multirow{2}{*}{Bound on \(d\)} &
\multirow{2}{*}{Estimated \(d\)} &
\multicolumn{2}{c}{Asymptotic coding gain} \\
\cline{3-8} \cline{11-12}
& & \(r=0\) & \(r=1\) & \(r=2\) & \(r=3\) & \(r=4\) & \(r=5\)
& & & \(G_\infty\) & \(G_\infty\) [dB] \\
\hline

\(1/2\) & 374 & 11 & 12 &   &   &   &   & \(d\le23\) & 23 & 11.5 & 10.6 \\
\(2/3\) & 375 & 11 & 8  & 4 &   &   &   & \(d\le12\) & 12 & 8.0  & 9.0  \\
\(3/4\) & 376 & 5  & 4  & 6 & 8 &   &   & \(d\le9\)  & 9  & 6.75  & 8.3  \\
\(4/5\) & 375 & 4  & 8  & 4 & 2 & 5 &   & \(d\le6\)  & 6  & 4.8  & 6.8  \\
\(5/6\) & 378 & 7  & 6  & 2 & 4 & 2 & 2 & \(d\le4\)  & 4  & 3.33  & 5.22  \\
\hline
\end{tabular}
\end{table*}

For fixed polynomial weight $\gamma$, the column weights satisfy $\sum_r W_{n_0}(r)=\gamma$. Hence, the minimum pairwise sum in \eqref{eq:distance_bound_general} is favored by distributing the polynomial exponents as uniformly as possible among the residue classes modulo $n_0$.

Table~\ref{tab:examples} summarizes Examples~\ref{ex:golo_base}--\ref{ex:separated_equal}. Slight perturbations in the exponents vectors can avoid null columns without introducing harmful short cycles, and the free distances estimated through the tool in~\cite{MacKayDist} attain the bound in all nondegenerate cases.

\begin{example}
\label{ex:large_golomb}
An optimal $23$-entry Golomb ruler of length $372$ is given by
\begin{align*}
\mathbf a=\bigl(&
0,3,7,17,61,66,91,99,114,159,171,199,\\ \nonumber
&
200,226,235,246,277,316,329,348,350,366,372
\bigr).
\end{align*}
Table~\ref{tab:large_golomb} reports the corresponding code parameters. Since all column weights are nonzero for $2\leq n_0\leq6$, all the corresponding constructions are nondegenerate. Moreover, because \(\mathbf a\) is a Golomb ruler, the parity-check matrices of the resulting \ac{LDPC} convolutional codes contain no length-$4$ cycles. The estimated free distance coincides with the bound for all considered rates. 
\end{example}

\section{Conclusion}\label{sec:concl}

We have investigated the construction of high-rate \ac{LDPC} convolutional codes obtained by periodic row elimination from the parity-check matrix of cyclic codes, yielding a family of convolutional codes with rate \(R=(n_0-1)/n_0\). We have characterized the survival of length-$4$ cycles in the resulting parity-check matrices, showing that the presence of short cycles can be predicted directly from the exponents vector associated with the parity-check polynomial.

We have also analyzed the role of the column-weight distribution of the
syndrome-former sub-matrix. In particular, we have derived explicit upper bounds on the free distance in terms of pairwise column-weight sums, showing that unbalanced distributions limit the achievable free distance. This provides a simple design criterion based on promoting uniformity of the column weights. The considered examples show that Golomb rulers guarantee the absence of length-$4$ cycles, but may produce all-zero columns for some values of $n_0$. Slight relaxations of the Golomb ruler constraint can improve the distribution of the polynomial exponents among the residue classes modulo $n_0$, thereby avoiding degenerate constructions without necessarily introducing harmful short cycles.

The evaluation of the finite-length and asymptotic decoding performance of these codes, i.e., evaluation of their block and bit error rates under iterative decoding, is left for future work.

\bibliographystyle{IEEEtran}
\bibliography{references}

\end{document}